\documentclass[11pt]{article}
\usepackage{amsmath,amssymb}
\usepackage{amsthm}
\usepackage{titlesec}
\usepackage{indentfirst} 

\usepackage{bbm}
\usepackage{mathrsfs}
\usepackage{amsfonts}

\usepackage{bm}
\usepackage{xy}      
\xyoption{all}
\usepackage{graphicx}

\theoremstyle{plain}
\newtheorem{thm}{Theorem}
\newtheorem{Pp}[thm]{Proposition}
\newtheorem{Co}[thm]{Corollary}
\newtheorem{Lm}[thm]{Lemma}
\newtheorem{Df}[thm]{Definition}

\newcommand{\ket}[1]		{{|#1\rangle}}

\newcommand{\qH}			{{\mathsf{H}}}

\newcommand{\cnot}			{{\textsf{c-Not}}}
\newcommand{\swap}		{{\textsf{Sw}}}
\newcommand{\qid}			{{\textsf{Id}}}
\newcommand{\had}			{{\textsf{H}}}
\newcommand{\phase}		{{\textsf{S}}}
\newcommand{\qpi}			{{\mathsf{\pi/8}}}

\newcommand{\qs}			{{\textsf{q}_\textsf{s}}}
\newcommand{\qh}			{{\textsf{q}_\textsf{h}}}
\newcommand{\blk}			{{\Box}}
\newcommand{\qA}			{{\mathbb{A}}}
\newcommand{\qU}			{{\mathbb{U}}}
\newcommand{\qD}			{{\mathbb{D}}}
\newcommand{\xL}			{{\textsf{L}}}
\newcommand{\xN}			{{\textsf{N}}}
\newcommand{\xR}			{{\textsf{R}}}

\newcommand{\cfg}[4]		{{ {#1} \; \overset{\overset{\text{{\normalsize $#2$}}}{\bigtriangledown}}{#3} \; {#4} }}

\newcommand{\dcBQP}		{{\mathsf{dcBQP}}}

\newcommand{\Proj}			{{\mathsf{Proj}}}

\newcommand{\smn}                  {{$s$-$m$-$n$}}

\newcommand{\udl}[1]                {\underline{#1}}

\newcommand{\rdl}[1]                {\underline{\underline{#1}}}

\newcommand{\qurtains}             {\hfill{QED}}

\newcommand{\pto}                   {\rightharpoonup}

\newcommand{\ignore}[1]             {{}}

\newcommand{\pnats}                 {{\mathbb{N}^+}}
\newcommand{\nats}                  {{\mathbb N}}

\newcommand{\Prob}                  {{\mathsf{Prob}}}

\newcommand{\dom}                   {\mathop{\textrm{dom}}}

\newcommand{\tab}                    {{\hspace*{5mm}}}

\renewcommand\thesection{\arabic{section}.\kern -.3em}
\renewcommand{\thesubsection}{\arabic{section}.\arabic{subsection}.\kern -.5em}

\begin{document}
\title{Quantum machines with classical control}
\author{Paulo Mateus$^{a}$\thanks{Corresponding author.} , Daowen Qiu$^{b}$,
\hskip 2mm Andr\'e Souto$^{c}$\footnote{An extended version of this paper was published in Festschrift volume following the conference honoring Am\'ilcar Sernadas by College Publications.}\\
\small{{\it $^a$Instituto de Telecomunica\c{c}\~{o}es, Departamento de Matem\'{a}tica, Instituto Superior T\'{e}cnico, }}\\
\small{{\it   University of Lisbon, Av. Rovisco Pais 1049-001, Lisbon, Portugal}}\\
\small{{\it $^b$Department of
Computer Science, Sun Yat-sen University, Guangzhou 510006, China}}\\
\small{{\it $^c$Departamento de Inform\'atica da Faculdade de Ci\^encias da Universidade Lisboa,
 }}\\
\small{{\it   LaSIGE, Faculdade de Ci\^encias, Universidade de Lisboa, Portugal}}\\
}
\date{ }
\maketitle

\begin{abstract}

Herein we survey the main results concerning quantum automata and machines with classical control. These machines were originally proposed by Sernadas {\em et al} in \cite{acs:pmat:yo:06}, during the FCT QuantLog project. 
First, we focus on the expressivity of quantum automata with both quantum and classical states. We revise the result obtained in \cite{dqiu:liu:pmat:acs:13} where it was proved that such automata are able to recognise, with exponentially less states than deterministic finite automata,  a family of regular languages that cannot be recognised by other types of quantum automata. 

Finally, we revise the concept of quantum Turing machine with classical control introduced in \cite{pmat:acs:asouto:14}. The novelty of these machines consists in the fact that their termination problem is completely deterministic, in opposition to other notions in the literature. Concretely, we revisit the result that such machines fulfil the $s$-$m$-$n$ property, while keeping the expressivity of a quantum model for computation.
\end{abstract}

\section{Introduction}

Quantum based machines were thought of by Feynman \cite{Fey82} when it became clear that quantum systems were hard to emulate with classical computers. The first notion of quantum Turing machine was devised by Deutsch \cite{Deu85}, and although it is sound and full working, it evolves using a quantum superposition of states. The main problem with Deutsch Turing machine is that checking its termination makes the machine evolution to collapse, interfering in this way with the quantum evolution itself.  Due to this issue, it is not obvious how one can extend the classical computability results, such as the $s$-$m$-$n$ property, for these quantum Turing machines. 

To avoid these complications, the community adopted other models of computation, such as quantum circuits, where it is relatively easy to present quantum algorithms. Moreover, these models were closer to what one was expecting, at that time, to be a physically implementable quantum computer. However, it soon became clear that implementing a full-fledged quantum computer was a long-term goal. For this reason, the community looked into restricted models of quantum computation, that required only a finite amount of memory -- quantum automata. Interestingly, although finite quantum automata (or measure once one-way quantum finite automata -- MO-1QFA) were able to accept some regular languages with exponentially less states,  MO-1QFA do not accept all regular languages. 

One of the main goals of the FCT QuantLog project was to understand if by endowing quantum systems or devices with classical control one could, in one hand, avoid the termination problem of quantum Turing machines, and on the other hand, extend the expressiveness of quantum automata, while keeping the exponential conciseness in the number of states. Some of these problems were introduced in the seminal paper by Sernadas {\em et al} \cite{acs:pmat:yo:06}.

During the QuantLog project two main results were attained along this line of research. Firstly, an extension of classical logic was proposed that was able to deal with quantum systems -- the exogenous quantum propositional logic (EQPL) \cite{pmat:acs:04a,pmat:acs:04b,pmat:acs:05,rchadha:pmat:acs:css:05,rchadha:pmat:acs:06b}. A second result, was an algorithm to minimize quantum automata \cite{PQL12,zhe:dqiu:gru:li:pmat:13}, which was an essential step to fully understand the exponential conciseness of quantum automata. Interestingly, the method also allowed to minimize constructively probabilistic automata (a problem that was open for more than 30 years) and that was previously characterized in terms of category theory in \cite{pmat:acs:css:99a,sch:pmat:99,chermida:pmat:00a,chermida:pmat:00b}. 

The previous results set the ground to show that quantum automata endowed with classical control recognise a family of regular languages that cannot be recognised by other types of quantum automata and moreover, with exponentially less states than deterministic finite automata \cite{dqiu:liu:pmat:acs:13}. Along the same line, by endowing classical Turing machines with a quantum tape, one can define a deterministic well behaved quantum Turing machine, where the usual classical theorems from computability can be derived \cite{pmat:acs:asouto:14}. Given the contribution of Am\'ilcar Sernadas to these two  elegant results, it is worthwhile to revisit them in this volume. The first set of results is presented in Section~2 and the second set of results is revised in Section~3.

\section{Quantum automata}\label{sec2}
A Measure Only One-way Quantum Finite Automaton (MO-1QFA) is defined as a quintuple
${\cal A}=(Q, \Sigma, |\psi_{0}\rangle,\{U(\sigma)\}_{\sigma\in\Sigma},Q_{acc})$,
where $Q$ is a set of finite states, $|\psi_{0}\rangle$ is the initial state
that is a superposition of the states in $Q$, $\Sigma$ is a finite
input alphabet, $U(\sigma)$ is a unitary matrix for each
$\sigma\in\Sigma$, and $Q_{acc}\subseteq Q$ is the
set of accepting states.

As usual, we identify $Q$ with an orthonormal base of a complex
Euclidean space and every state $q\in Q$ is identified with a
basis vector, denoted by Dirac symbol $|q\rangle$ (a column
vector), and $\langle q|$ is the conjugate transpose of
$|q\rangle$. We describe the computing process for any given input
string $x=\sigma_{1}\sigma_{2}\cdots\sigma_{m}\in\Sigma^{*}$. At
the beginning the machine ${\cal A}$ is in the initial state
$|\psi_{0}\rangle$, and upon reading $\sigma_{1}$, the
transformation $U(\sigma_{1})$ acts on $|\psi_{0}\rangle$. After
that, $U(\sigma_{1})|\psi_{0}\rangle$ becomes the current state
and the machine reads $\sigma_{2}$. The process continues until
the machine has read $\sigma_{m}$ ending in the state
$|\psi_{x}\rangle=U(\sigma_{m})U(\sigma_{m-1})\cdots
U(\sigma_{1})|\psi_{0}\rangle$. Finally, a measurement is
performed on $|\psi_{x}\rangle$ and the accepting probability
$p_{a}(x)$ is equal to
\[
p_{a}(x)=\langle\psi_{x}|P_{a}|\psi_{x}\rangle=\|P_{a}|\psi_{x}\rangle\|^{2}
\]
where $P_{a}=\sum_{q\in Q_{acc}}|q\rangle\langle q|$ is the
projection onto the subspace spanned by $\{|q\rangle: q\in
Q_{acc}\}$.

Now we further recall the definition of multi-letter QFA
\cite{BRS07}.

A $k$-letter 1QFA ${\cal A}$ is defined as a quintuple ${\cal
A}=(Q,\Sigma, |\psi_{0}\rangle, \nu,Q_{acc})$ where $Q$, $|\psi_{0}\rangle$, $\Sigma$, $Q_{acc}\subseteq Q$, are the same as those in MO-1QFA above, and $\nu$ is a function that assigns a unitary
transition matrix $U_{w}$ on $\mathbb{C}^{|Q|}$ for each string
$w\in (\{\Lambda\}\cup\Sigma)^{k}$, where $|Q|$ is the cardinality
of $Q$.

The computation of a $k$-letter 1QFA ${\cal A}$ works in the same
way as the computation of an MO-1QFA, except that it applies
unitary transformations corresponding not only to the last letter
but the last $k$ letters received. When
$k=1$, it is exactly an MO-1QFA as defined before. According
to \cite{BRS07,QY09}, the languages accepted by $k$-letter 1QFA  are a proper subset of regular languages for any $k$.

A Measure Many One-way Quantum Finite Automaton (MM-1QFA) is defined as a 6-tuple
${\cal A}=(Q,\Sigma,|\psi_{0}\rangle,\{U(\sigma)\}_{\sigma\in\Sigma\cup
\{\$\}},Q_{acc},Q_{rej})$, where $Q,Q_{acc}\subseteq
Q,|\psi_{0}\rangle,\Sigma,\{U(\sigma)\}_{\sigma\in\Sigma\cup
\{\$\}}$ are the same as those in an MO-1QFA defined above,
$Q_{rej}\subseteq Q$ represents the set of rejecting states, and
$\$\not\in\Sigma$ is a tape symbol denoting the right end-mark.
For any input string
$x=\sigma_{1}\sigma_{2}\cdots\sigma_{m}\in\Sigma^{*}$, the
computing process is similar to that of MO-1QFAs except that after
every transition, ${\cal A}$ measures its state with respect to the three
subspaces that are spanned by the three subsets $Q_{acc},
Q_{rej}$, and $Q_{non}$, respectively, where $Q_{non}=Q\setminus
(Q_{acc}\cup Q_{rej})$. In other words, the projection measurement
consists of $\{P_{a},P_{r},P_{n}\}$ where $P_{a}=\sum_{q\in
Q_{acc}}|q\rangle\langle q|$, $P_{r}=\sum_{q\in
Q_{rej}}|q\rangle\langle q|$, $P_{n}=\sum_{q\in Q\setminus
(Q_{acc}\cup Q_{rej})}|q\rangle\langle q|$. The machine stops
after the right end-mark $\$$ has been read. Of course, the
machine may also stop before reading  $\$$ if the current state, after
the machine reading some $\sigma_{i}$ $(1\leq i\leq m)$, does not
contain the states of $Q_{non}$. Since the measurement is
performed after each transition with the states of $Q_{non}$ being
preserved, the accepting probability $p_{a}(x)$ and the rejecting
probability $p_{r}(x)$ are given as follows (for convenience, we
denote $\$=\sigma_{m+1}$):
\[
p_{a}(x)=\sum_{k=1}^{m+1}\|P_{a}U(\sigma_{k})\prod_{i=1}^{k-1}(P_{n}U(\sigma_{i}))|\psi_{0}\rangle\|^{2},
\]
\[
p_{r}(x)=\sum_{k=1}^{m+1}\|P_{r}U(\sigma_{k})\prod_{i=1}^{k-1}(P_{n}U(\sigma_{i}))|\psi_{0}\rangle\|^{2}.
\]
Here we define $\prod_{i=1}^{n}A_i= A_nA_{n-1}\cdots A_1$.

Bertoni {\it et al} \cite{BMP03}
introduced a 1QFA, called 1QFACL that allows a more
general measurement than the previous  models. Similar to the case
in MM-1QFA, the state of this model can be observed at each step,
but an observable ${\cal O}$ is considered with a fixed, but
arbitrary, set of possible results ${\cal C}=\{c_1,\dots,c_n\}$,
without limit to $\{a,r,g\}$ as in MM-1QFA. The accepting
behavior in this model is also different from that of the previous
models. On any given input word $x$, the computation displays a
sequence $y\in {\cal C}^{*}$ of results of ${\cal O}$ with a
certain probability $p(y|x)$, and the computation is accepted if
and only if $y$ belongs to a fixed regular language ${\cal
L}\subseteq {\cal C}^{*}$.   Bertoni {\it et al}
\cite{BMP03} called  such a language ${\cal L}$ {\it control
language}.

More formally, given an input alphabet $\Sigma$ and the end-marker
  symbol $\$\notin\Sigma$, a 1QFACL over the working
  alphabet $\Gamma=\Sigma\cup\{\$\}$ is a five-tuple ${\cal
  M}=(Q, |\psi_{0}\rangle,\{U(\sigma)\}_{\sigma\in\Gamma},{\cal O},{\cal L})$, where
\begin{itemize}

\item $Q$, $|\psi_{0}\rangle$ and $U(\sigma)$ $(\sigma\in\Gamma)$ are defined
as in the case of MM-1QFA;

\item ${\cal O}$ is an observable with the set of possible results
${\cal C}=\{c_1,\dots,c_s\}$ and  the projector set
$\{P(c_i):i=1,\dots,s\}$ of which $P(c_i)$ denotes the projector
onto the eigenspace corresponding to $c_i$;

\item  ${\cal L}\subseteq{\cal C}^{*}$ is a regular language
(control language).

\end{itemize}

The input word $w$ to 1QFACL ${\cal M}$ is in the form:
$w\in\Sigma^{*}\$$, with symbol $\$$ denoting the end of a word.
Now, we define the behavior of ${\cal M}$ on word $x_1\dots
x_n\$$. The computation starts in the state $|\psi_{0}\rangle$, and then the
transformations associated with the symbols in the word  $x_1\dots
x_n\$$ are applied in succession. The transformation associated
with any symbol $\sigma\in\Gamma$ consists of two steps:
\begin{enumerate}
\item[1.] First, $U(\sigma)$ is applied to the current state
$|\phi\rangle$ of ${\cal M}$, yielding the new state
$|\phi^{'}\rangle=U(\sigma)|\phi\rangle$.

\item[2.] Second, the observable ${\cal O}$ is measured on
$|\phi^{'}\rangle$. According to quantum mechanics principle, this
measurement yields result $c_k$ with probability
$p_k=||P(c_k)|\phi^{'}\rangle||^2$, and the state of ${\cal M}$
collapses to $P(c_k)|\phi^{'}\rangle  /\sqrt{p_k}$.
\end{enumerate}

Thus, the computation on word $x_1\dots x_n\$$ leads to a sequence
$y_1\dots y_{n+1}\in {\cal C}^{*}$ with probability $p(y_1\dots
y_{n+1}|x_1\dots x_n\$)$ given by
\begin{equation}
p(y_1\dots y_{n+1}|x_1\dots
x_n\$)=\|  \prod^{n+1}_{i=1}P(y_i)    U(x_i)              |\psi_{0}\rangle\|^2,
\end{equation}
where we let $x_{n+1}=\$$ as stated before. A computation leading
to the word $y\in {\cal C}^{*}$ is said to be  accepted if $y\in
{\cal L}$. Otherwise, it is rejected. Hence, the accepting probability of 1QFACL ${\cal M}$ is
defined as:
\begin{equation}
{\cal P}_{\cal M}(x_1\dots x_n)=\sum_{y_1\dots y_{n+1}\in {\cal
L}}p(y_1\dots y_{n+1}|x_1\dots x_n\$).\label{f_CL}
\end{equation}

\subsection{One-way quantum automata together with classical states}

In the introduction we gave the motivation for introducing the new one-way quantum finite automata model, i.e., 1QFAC. We now define formally the  model.
To this end, we need the following notations. Given a finite set $B$, we denote by ${\cal H}(B)$ the Hilbert space freely generated by $B$. Furthermore, we denote by $I$ and $O$ the identity operator and zero operator on ${\cal H}(Q)$, respectively.

\begin{Df}\em \label{Df1}
A 1QFAC ${\cal A}$ is defined by a 9-tuple
\[
{\cal A}=(S,Q,\Sigma,\Gamma, s_{0},|\psi_{0}\rangle, \delta,
\mathbb{U}, {\cal M})
\]
where:
\begin{itemize}
\item $\Sigma$  is a finite set (the {\it input alphabet});

\item $\Gamma$  is a finite set (the {\it output alphabet});

\item $S$ is a finite set (the set of {\em classical states});

\item $Q$ is a finite set (the {\em quantum state basis});

\item $s_{0}$ is an element of $S$ (the {\em initial classical state});

\item $|\psi_{0}\rangle$ is a unit vector in the Hilbert space ${\cal H}(Q)$ (the {\em
initial quantum state});

\item $\delta: S\times \Sigma\rightarrow S$ is
a map (the {\em classical transition map});

\item $\mathbb{U}=\{U_{s\sigma}\}_{s\in S,\sigma\in \Sigma}$ where $U_{s\sigma}:{\cal
H}(Q)\rightarrow {\cal H}(Q)$ is a unitary operator for each $s$
and $\sigma$ (the {\em quantum transition operator} at $s$ and $\sigma$);

\item ${\cal M}=\{{\cal M}_s\}_{s\in S}$  where each ${\cal M}_s$ is a projective measurement over ${\cal H}(Q)$ with outcomes in $\Gamma$ (the {\em measurement operator at} $s$).
\end{itemize}
\end{Df}

Hence, each ${\cal M}_s= \{P_{s,\gamma}\}_{\gamma\in \Gamma}$ such that
$\sum_{\gamma\in \Gamma}P_{s,\gamma}=I$ and
$P_{s,\gamma}P_{s,\gamma'}=\left\{\begin{array}{ll}P_{s,\gamma},&
\gamma=\gamma',\\
O,& \gamma\not=\gamma'.
\end{array}
\right.$
Furthermore, if the machine is
in classical state $s$ and quantum state $|\psi\rangle$ after
reading the input string, then  $\|P_{s,\gamma}|\psi\rangle\|^{2}$
is the probability of the machine producing
outcome $\gamma$ on that input.

Note that the map
$\delta$ can be extended to a map $\delta^{*}:
\Sigma^{*}\rightarrow S$ as usual. That is,
$\delta^{*}(s,\epsilon)=s$; for any string $x\in\Sigma^{*}$ and
any $\sigma\in \Sigma$, $\delta^{*}(s,\sigma x)=
\delta^{*}(\delta(s,\sigma),x)$.

A specially interesting case of the above definition is when $\Gamma=\{a, r\}$,
    where $a$ denotes  {\em accepting} and $r$ denotes {\em rejecting}.
Then, ${\cal M}=\{\{P_{s,a},P_{s,r}\}:s\in S\}$ and, for each
$s\in S$, $P_{s,a}$ and $P_{s,r}$ are two projectors such that $P_{s,a}+P_{s,r}=I$ and $P_{s,a}P_{s,r}=
O$. In this case, ${\cal A}$ is an
acceptor of languages over $\Sigma$.

For the sake of convenience, we denote the map $\mu : \Sigma^{*}
\rightarrow S$, induced by $\delta$, as
$\mu(x)=\delta^{*}(s_{0},x)$ for any string $x\in\Sigma^{*}$.

We further describe the computing process of ${\cal
A}=(S,Q,\Sigma, s_{0},|\psi_{0}\rangle, \delta,
\mathbb{U},{\cal M})$ for input string
$x=\sigma_{1}\sigma_{2}\cdots\sigma_{m}$ where $\sigma_{i}\in
\Sigma$ for $i=1,2,\cdots,m$.

The machine ${\cal A}$ starts at the
initial classical state $s_{0}$ and initial quantum state
$|\psi_{0}\rangle$. On reading the first symbol $\sigma_{1}$ of the input string, the states of the machine change as follows: the classical state becomes
$\mu(\sigma_{1})$; the quantum state becomes $U_{s_{0}\sigma_{1}}|\psi_{0}\rangle$.
Afterward, on reading $\sigma_{2}$, the machine changes its classical state to $\mu(\sigma_{1}\sigma_{2})$ and its quantum state to the result of applying $U_{\mu(\sigma_{1})\sigma_{2}}$ to
$U_{s_{0}\sigma_{1}}|\psi_{0}\rangle$.

The process continues
similarly by reading $\sigma_{3},\sigma_{4},\cdots,\sigma_{m}$ in succession.
Therefore, after reading $\sigma_{m}$, the classical state becomes
$\mu(x)$ and the quantum state is as follows:
\begin{equation}
U_{\mu( \sigma_{1}\cdots \sigma_{m-2}\sigma_{m-1} )\sigma_{m}}U_{\mu(\sigma_{1}\cdots \sigma_{m-3}\sigma_{m-2})\sigma_{m-1}}\cdots U_{\mu(\sigma_{1})\sigma_{2}}
U_{s_{0}\sigma_{1}}|\psi_{0}\rangle.
\end{equation}

Let ${\cal U}(Q)$ be the set of unitary operators on Hilbert space
${\cal H}(Q)$. For the sake of convenience, we denote the map
$v:\Sigma^{*}\rightarrow {\cal U}(Q)$ as: $v(\epsilon)=I$ and
\begin{equation}
v(x)=
U_{\mu( \sigma_{1}\cdots \sigma_{m-2}\sigma_{m-1} )\sigma_{m}}U_{\mu(\sigma_{1}\cdots \sigma_{m-3}\sigma_{m-2})\sigma_{m-1}}\cdots U_{\mu(\sigma_{1})\sigma_{2}}
U_{s_{0}\sigma_{1}} \label{v}
\end{equation}
for $x=\sigma_{1}\sigma_{2}\cdots\sigma_{m}$ where $\sigma_{i}\in
\Sigma$ for $i=1,2,\cdots,m$, and $I$ denotes the identity
operator on ${\cal H}(Q)$, indicated as before.

By means of the denotations $\mu$ and $v$, for any input string
$x\in\Sigma^{*}$, after ${\cal A}$ reading $x$, the classical
state is $\mu(x)$ and the quantum states $v(x)|\psi_{0}\rangle$.

Finally,  the probability $\Prob_{{\cal A},\gamma}(x)$ of machine
${\cal A}$ producing result $\gamma$ on input $x$  is as follows:
\begin{equation}
\Prob_{{\cal A},\gamma}(x)= \|P_{\mu(x),\gamma}v(x)|\psi_{0}\rangle\|^{2}.
\end{equation}

In particular, when ${\cal A}$ is thought of as an acceptor of
languages over $\Sigma$ ($\Gamma=\{a,r\}$), we obtain the
probability $\Prob_{{\cal A},a}(x)$ for accepting $x$:
\begin{equation}
\Prob_{{\cal A},a}(x)= \|P_{\mu(x),a}v(x)|\psi_{0}\rangle\|^{2}.
\end{equation}

If a 1QFAC ${\cal A}$ has only one classical state,
then ${\cal A}$ reduces to an MO-1QFA \cite{MC00}. Therefore, the
set of languages accepted by  1QFAC with only one classical state
is a proper subset of regular languages, the languages
whose syntactic monoid is a group  \cite{BP02}. However, we revisit here the result obtained in \cite{dqiu:liu:pmat:acs:13} that 1QFAC
can accept all regular languages with no error.

\begin{Pp}\em  Let $\Sigma$ be a finite set. Then each regular language over  $\Sigma$ that is accepted by a minimal DFA of $k$ states
is also accepted by some 1QFAC with no error and with 1 quantum basis state and $k$ classical states.
\end{Pp}

\begin{proof}
Let $L\subseteq \Sigma^{*}$ be a regular language. Then there
exists a DFA
$M=(S,\Sigma,\delta,s_{0},F)$ accepting $L$, where, as usual, $S$
is a finite set of states, $s_{0}\in S$ is an initial state,
$F\subseteq Q$ is a set of accepting states, and $\delta: Q\times
\Sigma\rightarrow Q$ is the transition function. We construct a
1QFAC ${\cal A}=(S,Q,\Sigma,\Gamma, s_{0},|\psi_{0}\rangle,
\delta, \mathbb{U}, {\cal M})$ accepting $L$ without error, where
$S$,  $\Sigma$, $s_{0}$, and $\delta$ are the same as those in
$M$, and, in addition, $\Gamma=\{a,r\}$, $Q=\{0\}$,
$|\psi_{0}\rangle=|0\rangle$, $\mathbb{U}=\{U_{s\sigma}: s\in
S,\sigma\in\Sigma\}$ with $U_{s\sigma}=I$ for all $s\in S$ and
$\sigma\in\Sigma$, ${\cal M}=\{\{P_{s,a},P_{s,r}\}:s\in S\}$
assigned as: if $s\in F$, then $P_{s,a}=|0\rangle\langle 0|$ and
$P_{s,r}=O$ where $O$ denotes the zero operator as before; otherwise,
$P_{s,a}=O$ and $P_{s,r}=|0\rangle\langle 0|$.

By the above definition of 1QFAC ${\cal A}$, it is easy to check
that the language accepted by ${\cal A}$ with no error is exactly
$L$.
\end{proof}

Observe that for any regular language $L$ over $\{0,1\}$ accepted by a $k$ state DFA, it was proved that there exists a 1QFACL accepting $L$ with no error and with $3k$ classical states ($3k$ is the number of states of its minimal DFA accepting the control language)  and 3 quantum basis states \cite{MP06}. Here, for 1QFAC, we require only $k$ classical states and 1 quantum basis states. Therefore, in this case, 1QFAC have better state complexity than 1QFACL.

On the other hand, any language accepted by a 1QFAC
is regular. We can prove this result in detail, based on a well-know idea for one-way probabilistic automata by Rabin \cite{Rab63}, that was already applied for MM-1QFA by Kondacs and Watrous \cite{KW97} as well as for MO-1QFA by Brodsky and Pippenger \cite{BP02}. However, the process is much longer and further results are needed, since both classical and quantum states are involved in 1QFAC. Another possible approach is based on topological automata \cite{Boz03,Jea07}.  However, in next section we obtain this result while studying the state complexity of 1QFAC and so we postpone the proof of regularity to the next section.

\subsection{State complexity of 1QFAC}
State complexity of classical finite automata has been a hot research subject with important practical applications \cite{Yu98}. In this section,  we consider this problem for 1QFAC. First, we  prove a lower bound on the state complexity of 1QFAC which states that 1QFAC are at most exponentially more concise than DFA. Second, we show that our bound is tight by giving some languages that witness the exponential advantage of 1QFAC over DFA. Particularly, these languages can not be accepted by any MO-1QFA, MM-1QFA or multi-letter 1QFA.

Here we prove a lower bound for the state complexity of 1QFAC which states that 1QFAC are at most exponentially more concise than DFA. Also, we show that the languages accepted by 1QFAC with bounded error are regular. Some examples given in the next subsection shows that our lower bound is tight.

Given a 1QFAC ${\cal A}=(S,Q,\Sigma,\Gamma, s_{0},|\psi_{0}\rangle, \delta,
\mathbb{U}, {\cal M})$, we shall consider the triple $$({\cal H}, |\phi_0\rangle, \{M(\sigma): \sigma\in\Sigma\}, \{P_\gamma: \gamma\in\Gamma\})$$ 
where
\begin{itemize}
  \item  ${\cal H}={\cal H}(S)\otimes {\cal H}(Q)$;
   \item $|\phi_0\rangle=|s_0\rangle|\psi_0\rangle$;
  \item $M(\sigma)=\sum_{s\in S}|\delta(s,\sigma)\rangle\langle s|\otimes U_{s\sigma}$ for $\sigma\in\Sigma$;
  \item $P_\gamma=\sum_{s\in S}|s\rangle\langle s|\otimes P_{s\gamma}$ for each $\gamma\in\Gamma$.
\end{itemize}
It is easy to verify that \begin{equation}\Prob_{{\cal A},\gamma}(x)=\|P_\gamma M(x)|\phi_0\rangle\|^2 \label{Prob}\end{equation} for each $\gamma\in \Gamma$ and $x\in\Sigma^*$, where $M(x_1\cdots x_n)=M(x_n)\cdots M(x_1)$.
Furthermore, we let
\begin{equation} {\cal V}=\{|\phi_x\rangle: |\phi_x\rangle=M(x)|\phi_0\rangle, x\in\Sigma^*\}. \label{df-v}\end{equation}
Then we have the following result.
\begin{Lm}\label{Lm-v}\em It holds that
\begin{itemize}
  \item [(i)] each $|\phi\rangle\in{\cal V}$ has the form $|\phi\rangle=|s\rangle|\psi\rangle$ where $s\in S$ and $|\psi\rangle\in {\cal H}(Q)$;
  \item [(ii)] $\||\phi\rangle\|^2=1$ for all $|\phi\rangle\in{\cal V}$;
  \item [(iii)] $\|M(x)|\phi_1\rangle-M(x)|\phi_2\rangle\|\leq\sqrt{2}\||\phi_1\rangle-|\phi_2\rangle\|$ for all $x\in\Sigma^*$.
\end{itemize}
\end{Lm}
\begin{proof} Items (i) and (ii) are easy to be verified. In the following, we prove item (iii). Let $|\phi_i\rangle=|s_i\rangle|\psi_i\rangle$ and $|\phi'_i\rangle=M(x)|\phi_i\rangle=|s'_i\rangle|\psi'_i\rangle$ for $i=1,2$ and $x\in \Sigma^*$, where $s_i, s'_i\in S$ and $|\psi_i\rangle, |\psi'_i\rangle\in {\cal H}(Q)$. The  discussion is divided into two cases.

\noindent
Case (a): $|s_1\rangle=|s_2\rangle$. In this case it necessarily holds that $|s'_1\rangle=|s'_2\rangle$ and furthermore we have
\begin{equation}
\||\phi'_1\rangle-\phi'_2\rangle\|=\||\psi'_1\rangle-|\psi'_2\rangle\|=\||\psi_1\rangle-|\psi_2\rangle\|=\||\phi_1\rangle-|\phi_2\rangle\|,
\end{equation}
where the first and third equations hold because of $\||\alpha\rangle|\beta\rangle\|=\||\alpha\rangle\|.\||\beta\rangle\|$ and the second holds since $|\psi'_1\rangle$ and $|\psi'_2\rangle$ are obtained by performing the same unitary operation on $|\psi_1\rangle$ and $|\psi_2\rangle$, respectively.

\noindent
Case (b): $|s_1\rangle\neq|s_2\rangle$. First it holds that $\||\phi_1\rangle-|\phi_2\rangle\|=\sqrt{2}$. Indeed, let $|\psi_1\rangle=\sum_i\alpha_i|i\rangle$ and $|\psi_2\rangle=\sum_i\beta_i|i\rangle$. Then, we have
\begin{eqnarray}\||\phi_1\rangle-|\phi_2\rangle\|&=&\||s_1\rangle|\psi_1\rangle-|s_2\rangle|\psi_2\rangle\|\\
&=&\left\|\sum_i\alpha_i|s_1\rangle|i\rangle+ \sum_i(-\beta_i)|s_2\rangle|i\rangle\right\|\\
&=& \left(\sum_i |\alpha_i|^2+\sum_i |\beta_i|^2\right)^{\frac{1}{2}}\\
&=&\sqrt{\||\psi_1\rangle\|^2+\||\psi_1\rangle\|^2}\\
&=&\sqrt{2}.
\end{eqnarray}
Therefore, \begin{eqnarray}
\||\phi'_1\rangle-\phi'_2\rangle||&=&\||s'_1\rangle|\psi'_1\rangle-|s'_2\rangle|\psi'_2\rangle\|\\
&=&\left\{
     \begin{array}{ll}
       \|\psi'_1\rangle-|\psi'_2\rangle\|, & \hbox{if $s'_1= s'_2$;} \\
       \sqrt{2}, & \hbox{else.}
     \end{array}
   \right.
\end{eqnarray}
Note that $ \|\psi'_1\rangle-|\psi'_2\rangle\|\leq 2=\sqrt{2}\||\phi_1\rangle-|\phi_2\rangle\|$.

 In summary, item (iii) holds in any case.
\end{proof}

Next we present another lemma which is critical for obtaining the lower bound on 1QFAC.
\begin{Lm} \label{Lm-s}\em 
Let ${\cal V}_\theta\subseteq \mathbb{C}^n$ such that $\||\phi_1\rangle-|\phi_2\rangle\|\geq \theta$ for any two elements $|\phi_1\rangle,|\phi_2\rangle\in {\cal V}_\theta$. Then ${\cal V}_\theta$ is a finite set containing  $k(\theta)$ elements where $k(\theta)\leq(1+\frac{2}{\theta})^{2n}$.
\end{Lm}
\begin{proof} Arbitrarily choose an element $|\phi\rangle\in {\cal V}_\theta$. Let $U(|\phi\rangle, \frac{\theta}{2})=\{|\chi\rangle: \||\chi\rangle-|\phi\rangle\|\leq\frac{\theta}{2}\}$, i.e., a sphere centered at $|\phi\rangle$ with the radius $\frac{\theta}{2}$. Then  all these spheres do not intersect pairwise except for  their surface, and all of them are contained in a large sphere centered at $(0,0,\cdots,0)$ with the radius $1+\frac{\theta}{2}$. The volume of a sphere of a radius $r$ in $\mathbb{C}^n$ is $cr^{2n}$ where $c$ is a constant. Note that $\mathbb{C}^n$ is an $n$-dimensional complex space and each element from it can be represented by an element of $\mathbb{R}^{2n}$. Therefore, it holds that
\begin{equation}
k(\theta)\leq\frac{c(1+\frac{\theta}{2})^{2n}}{c(\frac{\theta}{2})^{2n}}=\left(1+\frac{2}{\theta}\right)^{2n}.
\end{equation}
\end{proof}
\noindent
Below we recall a result that will be used later on (c.f. Lemma 8 in \cite{Ya03} for a complete proof).
\begin{Lm}\label{Lm-p}\em 
For any two elements $|\phi\rangle, |\varphi\rangle\in \mathbb{C}^n$ with $\||\phi\rangle\|\leq c$ and $\||\varphi\rangle\|\leq c$, it holds that $\left | \|P|\phi\rangle\|^2- \|P|\varphi\rangle\|^2   \right|\leq c\||\phi\rangle-|\varphi\rangle\|$ where $P$ is a projective operator on $\mathbb{C}^n$.
\end{Lm}

Given a language $L\subseteq\Sigma^*$, define an equivalence relation ``$\equiv_L$'' as: for any  $x,y\in\Sigma^*$, $x\equiv_L y$ if
for any $z\in\Sigma^*$, $xz\in L$ iff $yz\in
L$. If $x,y$ do not satisfy the equivalence relation, we denote it by $x\not\equiv_L y$. Then  the set $\Sigma^*$ is partitioned into some equivalence classes by the equivalence relation ``$\equiv_L$''. In the following we recall a well-known result that will be used in the sequel.
\begin{Lm}[Myhill-Nerode theorem \cite{HU79}]\label{MN-Th}\em 
A language $L\subseteq\Sigma^*$ is regular iff the number of equivalence classes induced by the equivalence relation ``$\equiv_L$'' is finite. Furthermore, the number of  equivalence classes equals to the state number of the minimal DFA accepting $L$.
\end{Lm}

Now we are ready to present our main result.
\begin{thm} \label{bound}\em 
 If $L$ is accepted by a 1QFAC ${\cal M}$ with bounded error, then $L$ is regular and it holds that $kn=\Omega( log ~m)$ where $k$ and $n$ denote  numbers of classical states and quantum basis states of ${\cal M}$, respectively, and $m$ is the state number of the minimal DFA accepting $L$.
\end{thm}
\begin{proof}
 Let ${\cal V}'\subseteq {\cal V}$ (where ${\cal V}$ is given in Eq. (\ref{df-v})) satisfying  for  any two elements $|\phi_{x}\rangle, |\phi_{y}\rangle\in {\cal V}'$  it holds that  $|\phi_{x}\rangle\not=|\phi_{y}\rangle\Leftrightarrow x\not\equiv_L y$.
Then for two different elements $|\phi_{x}\rangle, |\phi_{y}\rangle\in {\cal V}'$ there exists $z\in\Sigma^*$ satisfying $xz\in L$ whereas $yz\not\in L$ (or $xz\not\in L$ whereas $yz\in L$). That is
\begin{eqnarray} \Prob_{{\cal A},a}(xz)&=&||P_a M(z)|\phi_x\rangle||^2\geq\lambda+\epsilon,\\
\Prob_{{\cal A},a}(yz)&=&||P_a M(z)|\phi_y\rangle||^2\leq\lambda-\epsilon\end{eqnarray}
for some $\lambda\in(0,1]$ and $\epsilon>0$.
Therefore we have
\begin{eqnarray} \sqrt{2}\||\phi_x\rangle-|\phi_y\rangle\|&\geq &\|M(z)|\phi_x\rangle-M(z)|\phi_y\rangle\|\\
&\geq &|\Prob_{{\cal A},a}(xz)-\Prob_{{\cal A},a}(yz)|\\
&\geq &2\epsilon\end{eqnarray}
where the first inequality follows from Lemma \ref{Lm-v} and the second follows from Lemma \ref{Lm-p}.
  In summary, we obtain that two different elements $|\phi_{x}\rangle$ and $|\phi_{y}\rangle$ from ${\cal V}'$ satisfy $\||\phi_x\rangle-|\phi_y\rangle\|\geq \sqrt{2}\epsilon.$  Therefore, according to Lemma \ref{Lm-s}, we have that the number $|{\cal V}'|$ of elements in ${\cal V}'$ satisfies $|{\cal V}'|\leq(1+\frac{\sqrt{2}}{\epsilon})^{2kn} $, which means that the number of equivalence classes induced by  the equivalence relation ``$\equiv_L$'' is upper bounded by $(1+\frac{\sqrt{2}}{\epsilon})^{2kn}$. Therefore, by Lemma \ref{MN-Th} we have completed the proof. \end{proof}

When the number of classical states equals one in a 1QFAC  ${\cal M}$,   ${\cal M}$ exactly reduces to an MO-1QFA. Therefore,
as a corollary, we can obtain a precise relationship between the numbers of states for MO-1QFA and DFA that was also derived by Ablayev and Gainutdinova~\cite{AG00}.

\begin{Co}\em 
If $L$ is accepted by an MO-1QFA ${\cal M}$ with bounded error, then $L$ is regular and it holds that $n=\Omega( \log m)$ where  $n$ denotes the number of quantum basis states of ${\cal M}$, and $m$ is the state number of the minimal DFA accepting $L$.

\end{Co}

\subsection{The lower bound is tight}

Although 1QFAC accept only regular languages as  DFA, 1QFAC can accept some languages with essentially less number of states than DFA and these languages  cannot  be accepted by any MO-1QFA or MM-1QFA or multi-letter 1QFA. In this section, our purpose is to prove these claims, and we also obtain that the lower bound in Theorem \ref{bound} is tight.

First, we establish a technical result concerning the acceptability by 1QFAC of languages resulting from set operations on languages accepted by MO-1QFA and by DFA.

\begin{Lm} \label{operation}\em 

Let $\Sigma$ be a finite alphabet. Suppose that the language $L_{1}$ over  $\Sigma$ is accepted by a minimal DFA with $n_{1}$ states and  the language $L_{2}$ over  $\Sigma$ is accepted by an MO-1QFA with $n_{2}$ quantum basis states with bounded error $\epsilon$. Then  the intersection $L_{1}\cap L_{2}$, union $L_{1}\cup L_{2}$, differences $L_{1}\setminus L_{2}$ and $L_{2}\setminus L_{1}$ can be accepted by some 1QFAC  with  $n_{1}$ classical states and $n_{2}$ quantum basis states with bounded error $\epsilon$.
\end{Lm}

\begin{proof} Let $A_1=(S,\Sigma,\delta,s_0,F)$ be a minimal DFA accepting $L_1$, and let $A_2=(Q, \Sigma, |\psi_{0}\rangle, \\ \{U(\sigma)\}_{\sigma\in\Sigma},Q_{acc})$ be an MO-1QFA accepting $L_2$, where $s_0\in S$ is the initial state, $\delta$ is the transition function, and $F\subseteq S$ is a finite subset denoting accepting states; the symbols in $A_2$ are the same as those in the definition of MO-1QFA as above.

Then by $A_1$ and $A_2$ we define a 1QFAC ${\cal A}=(S,Q,\Sigma,\Gamma, s_{0},|\psi_{0}\rangle, \delta,
\mathbb{U}, {\cal M})$  accepting $L_{1}\cap L_{2}$, where  $S,Q,\Sigma, s_{0},|\psi_{0}\rangle, \delta$ are the same as those in $A_1$ and $A_2$, $\Gamma=\{a,r\}$,  $\mathbb{U}=\{U_{s\sigma}=U(\sigma): s\in S, \sigma\in \Sigma\}$, and ${\cal M}=\{M_{s}:s\in S\}$ where $M_s=\{P_{s,a},P_{s,r} \}$ and
\[
P_{s,a}=\left\{\begin{array}{ll}
\sum_{p\in Q_{acc}}|p\rangle\langle p|, & s\in F;\\
O,& s\not\in F,
\end{array}
\right.
\] where $O$ denotes the zero operator, and
$P_{s,r}=I-P_{s,a}$ with $I$ being the identity operator.

According to the above definition of 1QFAC, we easily know that, for any string $x\in \Sigma^*$, if $x\in L_1$ then the accepting probability of
1QFAC ${\cal A}$ is equal to the accepting probability of MO-1QFA $A_2$; if $x\not\in L_1$ then the accepting probability of
1QFAC ${\cal A}$ is zero. So, 1QFAC ${\cal A}$ accepts the intersection $L_{1}\cap L_{2}$.

Similarly, we can construct the other three 1QFAC accepting the union $L_{1}\cup L_{2}$, differences $L_{1}\setminus L_{2}$, and $L_{2}\setminus L_{1}$, respectively. Indeed, we only need define different measurements in these 1QFAC. If we construct 1QFAC accepting $L_{1}\cup L_{2}$, then
\[P_{s,a}=\left\{\begin{array}{ll}
I,&  s\in F;\\
\sum_{p\in Q_{acc}}|p\rangle\langle p|, & s\not\in F.
\end{array}
\right.
\]
If we construct 1QFAC accepting $L_{1}\setminus L_{2}$, then
\[
P_{s,a}=\left\{\begin{array}{ll}
\sum_{p\in Q\setminus Q_{acc}}|p\rangle\langle p|, & s\in F;\\
O,& s\not\in F.
\end{array}
\right.
\]
If we construct 1QFAC accepting $L_{2}\setminus L_{1}$, then
\[
P_{s,a}=\left\{\begin{array}{ll}
\sum_{p\in Q\setminus Q_{acc}}|p\rangle\langle p|, & s\not\in F;\\
O,& s\in F.
\end{array}
\right.
\]

\end{proof}

Consider the regular language $$L^0(m)=\{w0: w\in \{0,1\}^*, |w0|=km,k=1,2,3,\cdots \}.$$  Clearly, the minimal classical DFA accepting $L^0(m)$ has  $m+1$ states, as depicted in Figure~\ref{fig:dfal(m)}.

\begin{figure}[htbp]%
$$\xymatrix  {
 *+++[o][F-]{q_0}\ar@<-1ex>[r]^{0,1}
&*+++[o][F-]{q_1}\ar@<-1ex>[r]^{0,1}
&*+++[o][F-]{q_2}\ar@<-1ex>[r]^{0,1}
&*\txt{ \ \ \ ... \ \ \ }\ar@<-1ex>[r]^{0,1}
&*++[o][F-]{q_{\textrm{\tiny $m\!\!-\!\!1$}}}\ar@/_2pc/@(ul,ur)[llll]_1\ar@<-1ex>[r]^{0}
&*+++[o][F=]{q_m}\ar@/^3pc/@(dr,dl)[llll]^{0,1}
}$$
\caption{DFA accepting $L^0(m)$.}%
\label{fig:dfal(m)}%
\end{figure}
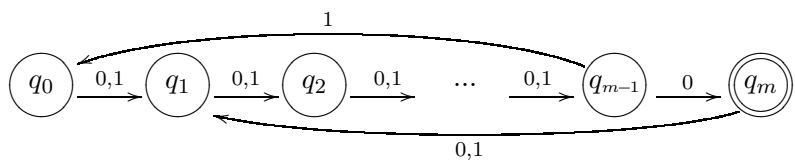

Indeed,  neither MO-1QFA nor MM-1QFA can accept
$L^0(m)$. We can easily verify this result by employing a lemma from \cite{BP02,GK02}. That is,

 \begin{Lm} [\cite{BP02,GK02}] \label{construction}\em  Let $L$ be a regular language, and let $M$ be its minimal DFA containing the construction in Figure 3, where states $p$ and $q$ are distinguishable (i.e., there exists a string  $z$ such that either $\delta(p,z)$ or  $\delta(q,z)$ is an accepting state). Then, $L$ can not be accepted by MM-1QFA.

 \end{Lm}

\begin{figure}[htbp]%
$$\xymatrix  {
 *+++[o][F-]{p}\ar@<-1ex>@/_/[rr]_{x}
&&*+++[o][F-]{q}\ar@<-1ex>@/_/[ll]_{y}\ar@(r,u)[]_{x}
}$$
\caption{Construction not accepted by an MM-1QFA.}%
\label{fig:construction}%
\end{figure}
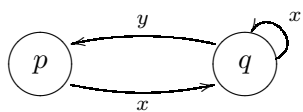

 \begin{Pp} \em Neither MO-1QFA nor MM-1QFA can accept
 $L^0(m)$.

 \end{Pp}

 \begin{proof}
 It suffices to show that no MM-1QFA can accept $L^0(m)$ since the  languages accepted by MO-1QFA are also accepted by MM-1QFA \cite{AF98,BP02,BMP03}.  By  Lemma \ref{construction},  we know that $L^0(m)$ can not be accepted by any MM-1QFA  since its minimal DFA (see Figure 2) contains such a construction: For example, we can take $p=q_0, q=q_m, x=0^m, y=0^{m-1}1, z=\epsilon$.
\end{proof}

In the following we recall a relevant result.

\begin{Pp}[\cite{AF98}] \label{Lp}\em Let the language $L_{p}=\{a^i: i \hskip 2mm  \textrm{is divisible by} \hskip 2mm p\}$ where $p$ is a prime number. Then for any $\varepsilon>0$, there exists an MM-1QFA with $O(\log(p))$ states such that for any $x\in L_{p}$, $x$ is accepted with no error, and the  probability for accepting $x\not\in  L_{p}$ is smaller than  $\varepsilon$.
 \end{Pp}

Indeed, from the proof of Proposition \ref{Lp} by \cite{AF98}, also as Ambainis and Freivalds pointed out in \cite{AF98} (before Section 2.2 in \cite{AF98}), Proposition \ref{Lp} holds for MO-1QFA as well.

Clearly, by the same technique used in the proof of Proposition \ref{Lp} \cite{AF98},  one can obtain that, by replacing $L_{p}$ with $L(m)=\{w: w\in \{0,1\}^*, |w|=km,k=1,2,3,\cdots \}$ with $m$ being a prime number, Proposition \ref{Lp} still holds (by viewing all input symbols in $\{0,1\}$ as $a$).
By combining Proposition \ref{Lp} with Lemma \ref{operation}, we have the following corollary.

\begin{Co}\em  Suppose that $m$ is a prime number. Then for any $\varepsilon>0$, there exists a 1QFAC with 2 classical states and $O(\log(m))$ quantum basis states such that for any $x\in L^0(m)$, $x$ is accepted with no error, and the probability for accepting $x\not\in  L^0(m)$ is smaller than  $\varepsilon$.
 \end{Co}
\begin{proof} Note  that  we have $$L^0(m)=L^0\cap L(m)$$ where  $L^0=\{w0: w\in \{0,1\}^*\}$ is accepted by a DFA (depicted in Figure \ref {fig:dfaw0}) with only two states and $L(m)$ can be accepted by an MO-1QFA with $O(\log(m))$  quantum basis states as shown in Proposition \ref{Lp}. Therefore, the result follows from Lemma \ref{operation}.
\end{proof}
\begin{figure}[htbp]%
$$\xymatrix  {
 *+++[o][F-]{q_0}\ar@<-1ex>@/_/[rr]_{0}\ar@(l,u)[]^{1}
&&*+++[o][F=]{q_1}\ar@<-1ex>@/_/[ll]_{1}\ar@(r,u)[]_{0}
}$$
\caption{DFA accepting $\{0,1\}^*0$.}%
\label{fig:dfaw0}%
\end{figure}
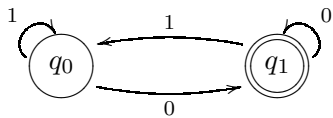

In summary, we have the following result.

 \begin{thm}\label{prop:accbe}\em  For any prime number $m\geq 2$, there exists a regular language $L^0(m)$ satisfying: (1) neither MO-1QFA nor MM-1QFA can accept $L^0(m)$; (2) the number of states in the minimal DFA accepting $L^0(m)$ is $m+1$; (3) for any  $\varepsilon>0$, there exists a 1QFAC with 2 classical states and $O(\log(m))$ quantum basis states such that for any $x\in L^0(m)$, $x$ is accepted with no error, and the probability for accepting $x\not\in  L^0(m)$ is smaller than  $\varepsilon$.

 \end{thm}

 From the above result (see (2) and (3)) it follows that the lower bound given in Theorem \ref{bound} is tight, that is, attainable.

One should ask at this point whether similar results can be established for multi-letter 1QFA as proposed by Belovs et al. \cite{BRS07}.

Recall that $1$-letter 1QFA is exactly an MO-1QFA. Any  given
$k$-letter QFA can be simulated by some $k+1$-letter QFA. However, Qiu and Yu \cite{QY09}
proved that the contrary does not hold.  Belovs et al.
\cite{BRS07} have already showed that $(a+b)^{*}b$ can be accepted
by a 2-letter QFA but, as proved in \cite{KW97}, it cannot be
accepted by any MM-1QFA. On the other hand,  $a^*b^*$ can be accepted by MM-1QFA \cite{AF98} but it can not be accepted by any multi-letter 1QFA \cite{QY09}, and furthermore, there exists a regular language that can not be accepted by any MM-1QFA or  multi-letter 1QFA \cite{QY09}.

Let $\Sigma$ be an alphabet. For string $z=z_1\cdots z_n\in\Sigma^*$, consider the regular language  $$L_z=\Sigma^*z_1\Sigma^* z_2\Sigma^*\cdots \Sigma^* z_n \Sigma^*.$$ $L_z$  belongs to piecewise testable set that was introduced by Simon \cite{Sim75} and studied in \cite{Per94}. Brodsky and Pippenger \cite{BP02} proved that $L_z$ can be accepted by an MM-1QFA with $2n+3$ states.

Consider the following regular language $L(m)=\{w: w\in \Sigma^*, |w|=km,k=1,2,\cdots\}$. Then the minimal DFA accepting $L_z$ needs $n+1$ states, and the minimal DFA accepting the intersection $L_z(m)$ of $L_z$ and $L(m)$  needs  $m
(n+1)$ states. We will prove that no multi-letter 1QFA can accept $L_z(m)$. Indeed,
the minimal DFA accepting $L_z(m)$ can be described by $A=(Q, \Sigma, \delta, q_0, F)$ where
$Q=\{S_{ij}: i=0, 1, \dots, n; j=1, 2, \dots, m\}$, $\Sigma=\{z_1, z_2, \dots, z_n\}$, $q_0=S_{01}$,
$F=\{S_{n1}\}$, and the transition function $\delta$ is defined as:
\begin{equation}\label{minimalDFA}
    \delta(S_{ij},\sigma)=\left\{
                            \begin{array}{ll}
                              S_{n,(j{\hskip -2mm} \mod m)+1}, & \textrm{if } i=n, \\
                              S_{i+1,(j{\hskip -2mm} \mod m)+1}, & \textrm{if } i\neq n \textrm{ and } \sigma=z_{i+1}, \\
                              S_{i,(j{\hskip -2mm} \mod m)+1}, &\textrm{if } i\neq n \textrm{ and } \sigma\neq z_{i+1}.
                            \end{array}
                          \right.
\end{equation}

The number of states of the minimal DFA accepting $L_z(m)$ is $m
(n+1)$.

For the sake of simplicity, we consider a special case: $m=2$, $n=1$, and $\Sigma=\{0,1\}$. Indeed, this case can also show the above problem as desired. So, we consider the following language:
\[
L_0(2)=\{w: w\in \{0,1\}^*0\{0,1\}^*, |w|=2k,k=1,2,\cdots\}.
\]

The minimal DFA accepting $L_0(2)$ above needs 4 states and its transition figure is depicted by Figure~\ref{fig:dfa0even} as follows.

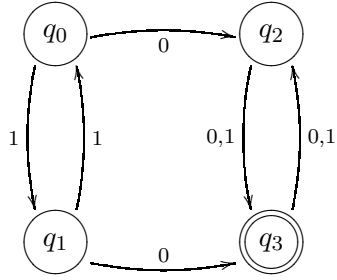
\begin{figure}[htbp]%
$$\xymatrix  {
 *+++[o][F-]{q_0}\ar@<-1ex>@/^/[rr]_{0}\ar@<-1ex>@/_/[dd]_{1}
&&*+++[o][F-]{q_2}\ar@<-1ex>@/_/[dd]_{0,1}\\\\
 *+++[o][F-]{q_1}\ar@<-1ex>@/_/[rr]^{0}\ar@<-1ex>@/_/[uu]_{1}
&&*+++[o][F=]{q_3}\ar@<-1ex>@/_/[uu]_{0,1}
}$$
\caption{DFA accepting $w\in \{0,1\}^*0\{0,1\}^*$ with $|w|$ even.}%
\label{fig:dfa0even}%
\end{figure}

We recall the definition of  F-construction and a proposition from \cite{BRS07}.

\begin{Df} [\cite{BRS07}]\em
A DFA  with state transition function $\delta$ is said to {\em contain
an F-construction} (see Figure~\ref{fig:fconstruction}) if there are non-empty words $t,z\in
\Sigma^{+}$ and two distinct states $q_{1},q_{2}\in Q$ such that
$\delta^{*}(q_{1},z)=\delta^{*}(q_{2},z)=q_{2}$,
$\delta^{*}(q_{1},t)=q_{1}$, $\delta^{*}(q_{2},t)=q_{2}$, where $\Sigma^+=\Sigma^*\backslash \{\epsilon\}$, $\epsilon$ denotes empty string.

\end{Df}

We can depict F-construction by Figure~\ref{fig:fconstruction}.

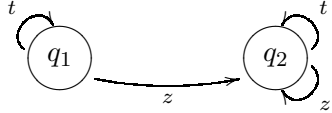
\begin{figure}[htbp]%
$$\xymatrix  {
 *+++[o][F-]{q_1}\ar@<-1ex>@/_/[rr]_{z}\ar@(l,u)[]^{t}
&&*+++[o][F-]{q_2}\ar@(r,u)[]_{t}\ar@(r,d)[]^{z}
}$$
\caption{F-Construction.}%
\label{fig:fconstruction}%
\end{figure}

\begin{Lm} [\cite{BRS07}]
A language $L$ can be accepted by a multi-letter 1QFA with bounded
error if and only if the minimal DFA of $L$ does not contain  any
F-construction.
\end{Lm}

In Figure~\ref{fig:dfa0even}, there are  an F-construction: For example,  we  consider $q_0$ and $q_3$, and  strings $00$ and $11$, from the above proposition which shows that no multi-letter 1QFA  can accept $L_0(2)$.

Therefore, similarly to Theorem \ref{prop:accbe}, we have:

\begin{thm}\em
 If  we have to restrict $m$ to be a prime number, then for any string $z$ with $|z|=n\geq 1$ there exists a regular language $L_z(m)$ that can not be accepted by any multi-letter 1QFA, but for every $\varepsilon$ there exists a 1QFAC ${\cal A}_m$  with $n+1$ classical states (independent of $m$) and $O(\log(m))$ quantum basis states such that if $x\in L_z(m)$, $x$ is accepted with no error, and the probability for accepting $x\not\in  L_z(m)$ is smaller than  $\varepsilon$. In contrast, the minimal DFA accepting $L_z(m)$ has $m(n+1)$ states.
\end{thm}

\section{Quantum Turing machines with classical control}\label{sec3}

Quantum Turing machines were proposed originally by Deutsch \cite{Deu85}. One of the main problems with Deutsch proposal is that it is hard to adapt and extend classical computability results using his notion of quantum machine, namely because states of the  Turing machine are quantum superpositions of classical states. 

To address this problem, \cite{per:jor:06} proposed a notion of quantum Turing machine where termination is similar to a probabilistic Turing machine. However, it is also not easy to derive computability results when the function computed by a Turing machine is a random variable. To address this issue, \cite{pmat:acs:asouto:14} proposed a notion of quantum Turing machine with deterministic control, which we revise here.

A {\it deterministic-control quantum Turing machine} (in short, {\it dcq Turing machine}) is 
a variant of a binary Turing machine with two tapes, one classical and the other with quantum contents, 
which are infinite in both directions.
Depending only on the state of the classical finite control automaton and the symbol being read by the classical head, the quantum head acts upon the quantum tape, 
a symbol can be written by the classical head, 
both heads can be moved independently of each other
and the state of the control automaton can be changed.

A computation ends if and when the control automaton reaches the halting state ($\qh$).
Notice that the contents of the quantum tape do not affect the computation flow, hence the deterministic control and, so, the deterministic halting criterion. In particular, the contents of the quantum tape do not influence at all if and when the computation ends.

The quantum head can act upon one or two consecutive qubits in the quantum tape. 
In the former case, it can apply any of the following operators to the qubit under the head: 
identity ($\qid$), Hadamard ($\had$), phase ($\phase$) and $\pi$ over 8 ($\qpi$).
In the latter case, the head acts on the qubit under it and the one immediately to the right 
by applying swap ($\swap$) or control-not ($\cnot$) 
with the control qubit being the qubit under the head.

Initially, 
the control automaton is in the starting state ($\qs$), 
the classical tape is filled with blanks (that is, with $\blk$'s) outside the finite input sequence $x$ of bits, 
the classical head is positioned over the rightmost blank before the input bits, 
the quantum tape contains three independent sequences of qubits 
-- an infinite sequence of $\ket0$'s followed by the finite input sequence $\ket\psi$ of possibly entangled qubits followed by an infinite sequence of $\ket0$'s,
and the quantum head is positioned over the rightmost $\ket0$ before the input qubits.
In this situation, we say that the machine starts with input
$(x,\ket\psi)$.

The control automaton is defined by the partial function
$$\delta: Q \times \qA \pto \qU \times \qD \times \qA \times \qD \times Q$$
where:
$Q$ is the finite set of control states containing at least the two distinct states $\qs$ and $\qh$ mentioned above;
$\qA$ is the alphabet composed of 0, 1 and $\blk$;
$\qU$ is the set $\{\qid,\had,\phase,\qpi,\swap,\cnot\}$ of primitive unitary operators
	that can be applied to the quantum tape;
and
$\qD$ is the set $\{\xL,\xN,\xR\}$ of possible head displacements --
one position to the left, none, and one position to the right.

For the sake of a simple halting criterion, we assume that 
$(\qh,a)\not\in\dom\delta$ for every $a\in\qA$ 
and
$(q,a)\in\dom\delta$ for every $a\in\qA$ and $q\neq\qh$.
Thus, as envisaged, the computation carried out by the machine does not terminate 
if and only if 
the halting state $\qh$ is not reached.

The machine evolves according to $\delta$ as expected:
$$\delta(q,a) = (U,d,a',d',q')$$
imposes that
if the machine is at state $q$ and reads $a$ on the classical tape,
then the machine
applies the unitary operator $U$ to the quantum tape,
displaces the quantum head according to $d$,
writes symbol $a'$ on the classical tape,
displaces the classical head according to $d'$,
and changes its control state to $q'$.

In short, by a dcq Turing machine we understand a pair $(Q,\delta)$ where $Q$ and $\delta$ are as above.


Concerning computations, the following terminology is useful.
The machine is said {\it to start from} $(x,\ket\psi)$ or to receive {\it input} $(x,\ket\psi)$ if:
(i)~the initial content of the classical tape is $x$ surrounded by blanks and the classical head is positioned in the rightmost blank before the classical input $x$;
(ii)~the initial content of the quantum tape is $\ket\psi$ surrounded by $\ket0$'s and the quantum head is positioned in the rightmost $\ket0$ before the quantum input $\ket\psi$.
Observe that the qubits containing the quantum input are not entangled with the other qubits of the quantum tape.
When the quantum tape is completely filled with $\ket0$'s we say that the quantum input is $\ket\varepsilon$.

Furthermore, the machine is said {\it to halt at} $(y,\ket\varphi)$ or to produce {\it output} $(y,\ket\varphi)$ if the computation terminates and:
(i)~the final content of the classical tape is $y$ surrounded by blanks and the classical head is positioned in the rightmost blank before the classical output $y$;
(ii)~the final content of the quantum tape is $\ket\varphi$ surrounded by $\ket0$'s and the quantum head is positioned in the rightmost $\ket0$ before the quantum output $\ket\varphi$.
In this situation we may write
$$M(x,\ket\psi) = (y,\ket\varphi).$$
Clearly, the qubits containing the quantum output are not entangled with the other qubits of the quantum tape.

For each $n\in\pnats$, denote by $\qH^n$ the Hilbert space of dimension $2^n$.
A unitary  operator 
$$U:\qH^n \to \qH^n$$ 
is said to be {\it dcq computable} if there is a dcq Turing machine $(Q,\delta)$ that,
for every unit vector $\ket\psi\in\qH^n$,
when starting from $(\varepsilon,\ket\psi)$ produces the quantum output $U\ket\psi$. 
Note that the final content of the classical tape is immaterial.

A (classical) problem
$$X \subseteq \{0,1\}^*$$
is said to be {\it dcq decidable} if 
there is a dcq Turing machine $(Q,\delta)$ that,
for every $x\in\{0,1\}^*$,
when starting from $(x,\ket\varepsilon)$ produces a quantum output $\ket\varphi$ such that:
$$
\begin{cases}
\Prob\left(\Proj_1\ket\varphi = 1\right) > 2/3 \;\text{ if }\; x\in X\\
\Prob\left(\Proj_1\ket\varphi = 0\right) > 2/3 \;\text{ if }\; x\not\in X
\end{cases}
$$
where $\Proj_1$ is the projective measurement defined by the operator
$$\begin{pmatrix} 0&0\\0&1 \end{pmatrix} \otimes \mathsf{ID}$$
using the adopted computational basis $\{\ket0,\ket1\}$,
with the first factor acting on the first qubit of the quantum output (the qubit immediately to the right of the quantum head) and the identity acting on the remaining qubits of the output. Clearly, the possible outcomes of the measurement  are the eigenvalues 0 and 1 of the defining operator.

Moreover, problem $X$ is said to be (time) {\it dcq bounded error quantum polynomial}, 
in short in $\dcBQP$, if 
there are polynomial $\xi\mapsto P(\xi)$ and 
a dcq Turing machine deciding $X$ that,
for each $x$, produces the output within $P(|x|)$ steps. In \cite{pmat:acs:asouto:14} it was established that the quantum computation concepts above coincide with those previously introduced in the literature using quantum circuits.

It is straightforward to see that dcq decidability coincides with the classical notion. It is enough to take into account that the dcq Turing machines can be emulated by classical Turing machines using a classical representation of the contents of the quantum tape that might be reached from $(x,\ket\varepsilon)$.

In the sequel we also need the following notion that capitalises on the fact that dcq Turing machines can work like classical machines by ignoring the quantum tape. A function
$$f:\{0,1\}^*\pto\{0,1\}^*$$
is said to be {\it classically dcq computable} if there is a dcq Turing machine $(Q,\delta)$ that,
for every $x\in\{0,1\}^*$,
when starting from input $(x,\ket\psi)$ produces the classical output $f(x)$ if $x \in \dom f$
and fails to halt with a meaningful classical output if $x \not\in \dom f$.

\begin{thm}[Polynomial translatability]\label{th:smn}\em\ 
There is a dcq Turing machine $T$ such that, 
for any dcq Turing machine $M=(Q,\delta)$,
there is a map 
$$s: \{0,1\}^* \to \{0,1\}^*$$
which is classically dcq computable in linear time and fulfils the following conditions:
\[\left\{\begin{array}{ll}
\forall\;p,x\in\{0,1\}^*, \ket\psi\in\qH^n,n\in\pnats 	& M(p \blk x, \ket\psi) = T(s(p) \blk x, \ket\psi)\\[2mm]
\exists\;c\in\nats\;\;\forall\;p\in\{0,1\}^* 			& |s(p)| \leq |p| + c.
\end{array}\right.\]
Moreover, there is a polynomial 
$(\xi_1,\xi_2,\xi_3) \mapsto P(\xi_1,\xi_2,\xi_3)$
such that 
if $M$ starting from $(p \blk x, \ket\psi)$ produces the output in $k$ steps
then $T$ produces the same output in at most
$$P(|p|+|x|,|Q|,k)$$
steps when starting from $(s(p) \blk x, \ket\psi)$.
\end{thm}
\begin{proof}
Without any of loss of generality assume that 
$$Q=\{q_0,q_1,\dots, q_\nu,q_{\nu+1}\}$$
with $\qs=q_0$ and $\qh=q_{\nu+1}$. Hence, $|Q|=\nu+2$.
Consider the map 
\[
s = p \mapsto \udl{\delta}\,111\,p: \{0,1\}^*\to\{0,1\}^*
\]  
where $\udl\delta$ encodes $\delta$ as follows:
$$\udl{\delta(q_0,0)}\,\udl{\delta(q_0,1)}\,\udl{\delta(q_0,\blk)}\dots
	\udl{\delta(q_i,0)}\,\udl{\delta(q_i,1)}\,\udl{\delta(q_i,\blk)}\dots
		\udl{\delta(q_\nu,0)}\,\udl{\delta(q_\nu,1)}\,\udl{\delta(q_\nu,\blk)}$$
with each
\[\udl{\delta(q,a)}=\udl{U}\,\udl{d}\,\udl{a'}\,\udl{d'}\,\udl{q'} \in\{0,1\}^*\]
where
\[\begin{array}{rcl}
\udl{U}&=&\left\{
			\begin{array}{rcl}
					000 & \text{ if } & U=\qid\\
					001 & \text{ if } & U=\had\\
					010 & \text{ if } & U=\phase\\
					011 & \text{ if } & U=\qpi\\
					100 & \text{ if } & U=\swap\\
					101 & \text{ if } & U=\cnot
			\end{array}
			\right.\\
\ \\
\udl{d}&=&\left\{
			\begin{array}{rcl}
					00 & \text{ if } & d=\xL\\
					01 & \text{ if } & d=\xN\\
					11 & \text{ if } & d=\xR
			\end{array}
			\right.\\
\ \\
\udl{a'}&=&\left\{
			\begin{array}{rcl}
					00 & \text{ if } & a'=0\\
					11 & \text{ if } & a'=1\\
					10 & \text{ if } & a'=\blk
			\end{array}
			\right.\\
\ \\
\udl{d'}&=&\left\{
			\begin{array}{rcl}
					00 & \text{ if } & d'=\xL\\
					01 & \text{ if } & d'=\xN\\
					11 & \text{ if } & d'=\xR
			\end{array}
			\right.\\
\ \\
\udl{q'}&=&1^{j+1}00 
\end{array} \]
assuming that $\delta(q,a) = (U,d,a',d',q')$ and $q'= q_j$.
Notice that one can identify in $s(p)$ the end of the encoding of $\udl{\delta}$ 
since each $\udl{\delta(q,a)}$
starts with $\udl U$
and the sequence $111$ does not encode any gate.
Clearly, as defined, $s$ can be dcq computed in linear time and fulfils the conditions in the statement of the theorem by taking $c=|\udl\delta|+3$.

It is necessary to encode in the classical tape of $T$ the current classical configuration of $M$ (composed of the current contents of the classical tape, the current position of the classical head and the current state of the control automaton). There is no need to encode the quantum configuration of $M$ since in a dcq Turing machine it does not affect its transitions. In due course, when explaining how $M$ computations are emulated by $T$ computations, we shall see how quantum configurations of $T$ are made to follow those of $M$. The following notation becomes handy for describing classical configurations of dcq Turing machines.

We write
$$\cfg{w}{q}{a}{w'}$$
for stating that the machine in hand is at state $q$, 
its classical head is over a tape cell containing symbol $a$,
with the finite sequence $w$ of symbols to the left of the head,
with the finite sequence $w'$ of symbols to the right of the head,
and with the rest of the classical tape filled with blanks.

Before describing how a classical configuration of
$M$ is encoded in $T$ we need to introduce some additional notation. Recall that a symbol $a\in\{0,1,\blk\}$ is encoded as
\[
\udl{a}=\left\{
			\begin{array}{rcl}
					00 & \text{ if } & a=0\\
					11 & \text{ if } & a=1\\
					10 & \text{ if } & a=\blk.
			\end{array}
			\right.
\]
We denote by $\udl{a}_1$ and $\udl{a}_2$ the first and second bit of $\udl{a}$, respectively. The reverse encoding of $a$ is $\udl{\udl{a}}=\udl{a}_2\udl{a}_1$. Given a string $w=w_1\dots w_m\in \{0,1,\blk\}^*$, we denote its encoding by $\udl{w}=\udl{w_1}\dots \udl{w_m}$ and its reverse encoding by $\udl{\udl{w}}=\udl{\udl{w_m}}\dots \udl{\udl{w_1}}$.

The classical configuration
 \[
 \cfg{w}{q_i}{a}{w'} 
\]
of $M$ should be encoded as the following classical configuration of $T$ 
\noindent
\[
\cfg
	{\udl{w}\; \blk\; \underbrace{\underbrace{1\dots 1}_{\nu+i-1}\; \blk\;
		\underbrace{1\dots 1}_{i+1}}_{q_i}} 
	{q'}
	{\blk} 
	{\udl{\delta} \; 111\; \rdl{w'} \;{\udl{a}}_2 {\udl{a}}_1 }
\]
where $q'$ is a state of $T$ representing the stage where the machine is able to start emulating a transition of $M$. As we shall see later, whenever a transition of $M$ has just been emulated by $T$ and the resulting state is not the halting state of $M$, $T$ is at state $q'$.

The initial classical configuration of $M$ is
\[
\cfg
	{}
	{q_0}{\blk}
	{p\;\blk\;x}
\]
and, moreover, the initial classical configuration of $T$ is
\[
\cfg
	{}
	{q'_0}
	{\blk}
	{\underbrace{\udl{\delta}\;111\;{p}}_{s(p)}\; \blk \; x},
\]
where $q'_0$ is the initial state of $T$.
The objective of this stage is to change the initial classical configuration of $T$
to the encoding of the initial configuration of $M$, as described before, that is:
\[
\cfg
	{\underbrace{1\dots 1\; \blk\; 1}_{\text{encoding of } q_0} }
	{q'}
	{\blk}
	{\udl{\delta} \;111\; \rdl{\blk p \blk x} }.
\]
Writing the encoding of $q_0$ can be done straightforwardly in 
$O(k)$ steps. 
It remains to  describe how to encode $\blk p\blk x$ in reverse order within 
$O((|p|+|x|)^2)$ steps, keeping $\udl \delta\;111$ unchanged:

\begin{enumerate}

\item{\bf Encoding of $x$}:
Recall that 
$x=x_1\dots x_m$ 
has no blanks, and therefore
$\udl x=x_1x_1\dots x_mx_m$. 
The idea is to 
	shift $x_2\dots x_m$ to the right, 
	duplicate $x_1$ in the vacated cell and 
	then iterate this process to $x_2\dots x_m$.
First, the head moves on top of $x_2$ and 
copies the contents of $x_2\dots x_m$ one cell to the right, leaving the original cell of $x_2$ with a blank. 
Then the head moves back to $x_1$ and 
copies its contents to the cell on its right. 
The process is iterated for $x_2\dots x_m$ until the last symbol of $x$ is reached.  
Since shifting to the right the contents of $m$ cells, leaving the first one blank, 
can be done with a linear number of steps in $m$, 
this operation takes a quadratic number of steps on the size of $x$.

\item{\bf Encoding the  $\blk$ in  $p\blk x$}:
First, 
	the encoding of $x$ is shifted one cell to the right and then, 
	the head is moved back to the top of the first two blanks separating ${p}$ and $\udl{x}$. 
Finally, the head replaces the two blanks by $10$ and it is parked in the $1$.
Note that this can be done with a linear number of steps on the size of $x$, and moreover, the encoding of $\blk x$ has no blanks.

\item{\bf Encoding of $p$}:
Let $l$ be the size of $p$. The encoding of $p$ is similar to the encoding of $x$. 
First the encoding of $\blk x$ is shifted three cells to the right.
Then the head of the machine is moved to the beginning of $p$. Notice that the machine can identify it as the first cell on the right of $\udl\delta\;111$. 
Next, $p$ is shifted one cell to the right (which leaves two blanks before  $\blk x$) and the head of the machine is moved to the cell containing $p_l$. 
The machine copies $p_l$ to the two cell immediately on its right and writes $\blk$ in the original cell. After these steps, the content of the classical tape is:
\[
\cfg
	{\udl{\delta} \;111\; \blk\; p_1\dots p_{l-1}}
	{q'}
	{\blk}
	{\udl{p_l\blk x} }.
\]
Next $\udl{p_l\blk x}$ is shifted one cell to the right and the process of writing the encoding of $p_{l}$ in the tape is repeated for $p_{l-1}$, $p_{l-2}$, \dots until $p_1$. The end of this construction is reached whenever the symbol $\blk$ is placed after $\udl \delta\;111$ is read.  
Finally, $\udl{p\;\blk\; x}$ is shifted two cells to the left.

\item{\bf Reversing the encoding of $p\blk x$}:
Assume that 
$\udl{p\blk x}=y_1\dots y_{m}$ 
(with  $m=|\udl{p}|+|\udl{x}|+2$) 
is the contents of the cells containing the encoding of 
$p\blk x$. 
The objective is to replace $y_1\dots y_{m}$ by 
$y_m\dots y_1$. 
First, the cell containing $y_1$ is replaced by a blank and $y_1$ is copied to the right cell of $y_{m}$. 
Second, the sequence 
$y_2\dots y_{m}$ is shifted one cell to the left. 
This process is repeated with 
$y_2\dots y_{m}$ 
in such a way that $y_2$ is copied to the left of $y_1$ and until the contents of the tape is 
$\blk y_m\dots y_1$.
Finally, the blank symbol is removed when 
$y_m\dots y_1$ 
is shifted one cell to the left. 
Observe that the operations leading to 
$\blk y_m\dots y_1$,
take $O(m^2)$ steps. Moreover, the final shift is linear, and so the overall stage takes a quadratic number of steps.

\item{\bf Placing the reverse encoding of a blank at the right end}:
The head is moved to the right until the first blank is found. Then, the head writes a $0$ and moves one cell to the right, where it writes a $1$. Finally, the head is moved to the left until the first blank is found. 

\end{enumerate}

It is straightforward to check that 
the overall cost of these operations is quadratic on $|p|+|x|$ 
and that the five stages above require just a constant number of states in $T$ 
(that is, the number of states does not depend on $p$ and $x$).

Next, we describe the steps needed to emulate in $T$ one step by $M$. 
Assume that the transition to be emulated is 
$\delta(q_i,a)=(U,d,a',d',q_j)$ and that 
$T$ is at the following classical configuration:

\[
\cfg
	{\udl{w}\; \underbrace{\underbrace{1\cdots 1}_{\nu-i-1} \;\blk\; \underbrace{1\cdots 1}_{i+1}}_{\text{encoding of } q_i} }
	{q'}
	{\blk}
	{\udl{\delta} \; 111 \; \rdl{w'}\; \udl{a_2}\; \udl{a_1}}.
\]

\noindent
The objective is to set $T$ at the following classical configuration
\[
\cfg
	{\udl{w''}\; \underbrace{\underbrace{1\cdots 1}_{\nu-j-1} \;\blk\; 
			\underbrace{1\cdots 1}_{j+1}}_{\text{encoding of } q_j} }
	{q'}
	{\blk}
	{\udl{\delta} \; 111 \; \rdl{ w''' }}
\]

\noindent
where, depending on the move of the classical head, three cases may occur:
\begin{itemize}
\item if $d'=N$ then
		$w=w''$ and 
		$w'''=a'w$;
\item if $d'=L$ then 
		$w''=w_1\dots{w}_{|w|-1}$ and 
		$w'''={w_{|w|} a' w'}$;
\item if $d'=R$ then 
		$w''={wa'}$ and
		$w'''=w'_2\dots w'_{|w'|}$.
\end{itemize}

\noindent
Machine $T$ performs the emulation of $\delta(q_i,a)$ as follows:

\begin{enumerate}

\item {\bf Identifying the value $a$}:
The head of the classical tape of $T$ is moved to $\udl{a}_1$ which is the rightmost cell that is not blank. 
The head reads the contents of that cell and the contents of the cell on its left, which has $\udl{a}_2$, and goes to a different state of $T$ depending on the value $a$.
The cost of this operation is linear in the number of states of $M$ and on the space used by $M$.

\item{\bf Parking the head at the encoding of $\delta(q_i,a)$ in $\udl{\delta}$}:
First the head is moved to the cell containing the rightmost $1$ of the encoding of $q_i$. 
Notice that such encoding has at least one $1$ to the right of the blank.  
Since the head starts from  position $\udl{a}_2$, such $1$ is on the left to the first blank that the head finds while reading the classical tape from right to the left. So, 
this operation is at most linear in the size of the space used by $M$. 
Recall that the encoding $\udl{\delta}$ of $\delta$ is as follows:
$$\textrm{\small$
	\udl{\delta(q_0,0)}\udl{\delta(q_0,1)}\udl{\delta(q_0,\blk)}\dots
		\udl{\delta(q_i,0)}\udl{\delta(q_i,1)}\udl{\delta(q_i,\blk)}\dots
			\udl{\delta(q_\nu,0)}\udl{\delta(q_\nu,1)}\udl{\delta(q_\nu,\blk)}$.}$$

Moreover, 
each $\udl{\delta(q_i,a)}$ ends with $00$ and starts 
with nine cells corresponding to $\udl{U}\cdot\udl{d}\cdot\udl{a'}\cdot\udl{d'}$ 
and a sequence of 1's, encoding the resulting state of that transition. 
This stage consists in a loop with progress variable, say $r$, starting from  $r=1$ until $r=i+1$. The goal of the loop is to replace the $r$ rightmost 1's of encoding of $q_i$ by 0's 
while the $00$, at the end of $\udl{\delta(q_{r-1},\blk)}$, are replaced by $\blk\blk$. 
The end of the loop $r=i+1$, is detected when a blank symbol is read in the encoding of $q_i$. 
For each value of $r$ we keep only a pair of  
$\blk\blk$ in $\udl{\delta}$:
those at the end of 
$\udl{\delta(q_{r-1},\blk)}$. 
When $r=i+1$, the encoding of $\delta(q_{i},\blk)$ is marked in $\delta$ with $\blk\blk$, 
and so, it remains to park the head in the first cell of the encoding of $\delta(q_i,a)$. 
This movement can be achieved taking into account the symbol $a$ read in the previous stage.
Observe that all the operations performed in this stage depend 
linearly on the space used by $M$ (on the right of its classical head) and
quadratically on the number of states of $M$.

\item{\bf Identifying and applying $U$}:
Using the first three cells of $\udl{\delta(q_i,a)}$, 
the machine $T$ identifies the unitary transformation and 
applies it to its own quantum tape. 

\item{\bf  Performing the $d$-move of the quantum head}:
Using the fourth and fifth cells of  $\udl{\delta(q_i,a)}$, 
$T$ identifies the movement of the quantum head and 
operates accordingly on its own quantum head.

\item{\bf Identifying and writing $a'$}:
Using the sixth and seventh cells of $\udl{\delta(q_i,a)}$, 
$T$ identifies the encoding of the symbol $a'$ to be written under $\udl{a}_2\udl{a}_1$. 
The encoding of $a'$ in $\udl{\delta(q_i,a)}$ is marked with two blanks and $
\udl{a'}$ is copied in reversed order to $\udl{a}_2\udl{a}_1$,
which are the two rightmost non-blank cells. 
After completing the last operation, the head  returns to the original position and restores $\udl{a'}$ in $\udl{\delta(q_i,a)}$. 
Notice that the operations of this stage can be done in a 
linear number of steps on the space used by $M$ the input and
linearly in the number of states.

\item{\bf Performing the $d'$-move of the classical head}:
The ninth and tenth cells of  $\udl{\delta(q_i,a)}$ store the movement of the classical head. 
If $d'=N$ nothing has to be done. 
W.l.o.g. assume $d'=R$. 
	First, the encoding of $d'$ is marked with two blanks. 
	Then the rightmost non-blank cells have to be copied (in reverse order) to the left of the leftmost non-blank cells. 
Clearly the rightmost non-blank cells have to be replaced by two blanks if $|w'|>0$, 
and have to be replaced by $01$ (the reverse encoding of a blank) if $|w'|=0$. 
{\it Mutatis mutandis} if $d'=L$. 
This stage can be done in a linear number of steps 
on the space in the classical tape used by $M$. 

\item{\bf Updating the emulated state to $q_j$}:
Assume that $q_j$ is not the halting state and 
recall that $q_j$ is encoded as $1^{j+1}00$ at the rightmost part of $\udl{\delta(q_i,a)}$. 
The idea is to update the emulated state $q_i$ to $q_j$ 
by replacing each 1 in $1^{j+1}00$ at $\udl{\delta(q_i,a)}$ by a $\blk$ 
while updating the cells used to encode the current state of $M$. 
Given Stage 2, the cells used to encode $q_i$ contain $1^{\nu-i-1}\blk 0^{i+1}$. 
If $j\leq i$ then 
	we replace $j+1$ rightmost 0's by 1's and then place a $\blk$ left to them. 
If $j>i$, then 
	the $i+1$ rightmost  0's are replaced by 1's and after, 
	the blank has to be carried to the left while being replaced by a 1, 
	until there are $(j+1)$ 1's. 
This process ends when all 1's in $1^{j+1}00$ have been replaced by blanks. 
After the cells encoding the emulated state are updated, 
the encoding of $q_j$ in $\udl{\delta(q_i,a)}$ is restored, 
by replacing the blanks by 1's. 
This stage does not depend on the input of $M$, but only 
quadratically in the number of states of $M$.
If $q_j$ is the halting state, 
we have to restore the contents of the classical tape to 
$wa'w'$, 
with the head positioned over $a'$. 
This corresponds to inverting the process used to prepare the initial configuration, erasing the encoding of $q_i$ and $\delta$. Such stage can be done in a number of steps 
quadratic to the space used by $M$ and
linearly in the number of states of $M$.
\end{enumerate}
Finally, the overall emulation is polynomial (in fact quadratic) on 
$|p|+|x|$, $\nu$ and $k$ since the space used by $M$ is bounded by $k$.
\end{proof}


A machine $T$ fulfilling the conditions of Theorem~\ref{th:smn} is said to enjoy the \smn\ property.
Any such machine is universal as shown in the next result.

\begin{thm}[Polynomial universality]\label{cor:univ}\em\ Let $T$ be a dcq Turing machine enjoying the \smn\ property.
Then, for any dcq Turing machine $M=(Q,\delta)$,
there is $p \in \{0,1\}^*$
such that
$$M(x, \ket\psi) = T(p \blk x, \ket\psi) \tab\forall\;x\in\{0,1\}^*,\ket{\psi}\in\qH^n,n\in\pnats.$$
Moreover, there is a polynomial 
$(\xi_1,\xi_2,\xi_3) \mapsto P(\xi_1,\xi_2,\xi_3)$
such that 
if $M$, when starting from $(x, \ket\psi)$, produces the output in $k$ steps
then $T$ produces the same output in at most
$$P(|x|,|Q|,k)$$
steps when starting from $(p \blk x, \ket\psi)$.
\end{thm}
\begin{proof} Consider $M'$ such that
$M'(\varepsilon \blk x,\ket\psi) = M(x,\ket\psi)$.
By applying Theorem~\ref{th:smn} to $M'$ and choosing $p=s(\varepsilon)$ the result follows.
\end{proof}

\subsubsection*{Acknowledgments}
This work is supported in part by the National
Natural Science Foundation (Nos. 61272058, 61073054, 60873055, 61100001), the Natural Science Foundation of Guangdong Province of China (No. 10251027501000004),   the Specialized Research Fund for the Doctoral Program of Higher Education of China (Nos. 20100171110042, 20100171120051), the Fundamental Research Funds for the Central Universities (No. 11lgpy36), and the project of  SQIG at IT, funded by FCT and EU FEDER projects QSec PTDC/EIA/67661/2006, AMDSC UTAustin/MAT/0057/2008, NoE Euro-NF,
and IT Project QuantTel, FCT project PTDC/EEA-TEL/103402/2008
QuantPrivTel, FCT PEst-OE/EEI/LA0008/2013.

The authors also acknowledge IT project QbigD funded by FCT PEst-OE/EEI/LA0008/2013, the Confident project
PTDC/EEI-CTP/4503/2014 and the support of LaSIGE Research Unit, ref. UID/CEC/00408/2013.

\end{document}